\documentclass[aps,pra,10pt,notitlepage,twocolumn]{revtex4-2}
\usepackage{amsfonts,amssymb,amsmath}
\usepackage{lmodern}
\usepackage[]{graphics,graphicx}            
\usepackage{amsthm}
\usepackage{enumitem,mathtools}
\usepackage{physics}
\usepackage[skins,breakable]{tcolorbox}
\usepackage{tikz}
\usetikzlibrary{decorations.pathmorphing,patterns,decorations.markings,matrix,quantikz2,backgrounds,calc}
\usepackage{adjustbox}
\usepackage{xstring}
\usepackage{ifthen}
\usepackage{listings}
    \usepackage{hyperref}
\usepackage[capitalize]{cleveref}

\usepackage[english]{babel}

\makeatletter
\def\bbl@set@language#1{%
  \edef\languagename{%
    \ifnum\escapechar=\expandafter`\string#1\@empty
    \else\string#1\@empty\fi}%
  \@ifundefined{babel@language@alias@\languagename}{}{%
    \edef\languagename{\@nameuse{babel@language@alias@\languagename}}%
  }%
  \select@language{\languagename}%
  \expandafter\ifx\csname date\languagename\endcsname\relax\else
    \if@filesw
      \protected@write\@auxout{}{\string\select@language{\languagename}}%
      \bbl@for\bbl@tempa\BabelContentsFiles{%
        \addtocontents{\bbl@tempa}{\xstring\select@language{\languagename}}}%
      \bbl@usehooks{write}{}%
    \fi
  \fi}
\newcommand{\DeclareLanguageAlias}[2]{%
  \global\@namedef{babel@language@alias@#1}{#2}%
}
\makeatother

\DeclareLanguageAlias{en}{english}

\newtheorem{theorem}{Theorem}

\newtheorem{definition}{Definition}
\newtheorem{lemma}{Lemma}

\newenvironment{proof*}[1][\proofname]{%
  
  \begin{proof}[#1]}{\end{proof}}

\renewcommand{\epsilon}{\varepsilon}
\crefformat{footnote}{#2\footnotemark[#1]#3}

\newtcolorbox[auto counter]{example}{enhanced,fonttitle=\sffamily\bfseries\large,valign=center
,drop fuzzy shadow,title={Example \thetcbcounter},breakable}

\DeclareFontFamily{U}{bbold}{}
\DeclareFontShape{U}{bbold}{m}{n}
 {
  <-5.5> s*[1.069] bbold5
  <5.5-6.5> s*[1.069] bbold6
  <6.5-7.5> s*[1.069] bbold7
  <7.5-8.5> s*[1.069] bbold8
  <8.5-9.5> s*[1.069] bbold9
  <9.5-11> s*[1.069] bbold10
  <11-15> s*[1.069] bbold12
  <15-> s*[1.069] bbold17
 }{}

\DeclareRobustCommand{\identity}{%
  \text{\usefont{U}{bbold}{m}{n}1}%
}

\begin{document}

\title{Demystifying the Balanced Product Code: A Review}
\date{\today}
\author{Heather \surname{Leitch}}
\author{Alastair \surname{Kay}}
\email{alastair.kay@rhul.ac.uk}
\affiliation{Royal Holloway University of London, Egham, Surrey, TW20 0EX, UK}

\begin{abstract}
The discovery of the family of balanced product codes \cite{breuckmann2021} was pivotal in the subsequent development of `good' low density quantum error correcting codes that have optimal scaling of the key parameters of distance and storage density. We review this family, giving a completely different presentation to the original, minimising the abstraction and technicalities wherever possible. The target audience is anyone familiar with the stabilizer formalism for error correction wanting to understand how parity-check matrices can be constructed for high storage density quantum codes.
\end{abstract}
\maketitle

Quantum computing sits at an inspiring threshold. There are exciting algorithms that we want to implement, such as Shor's algorithm \cite{shor1997a}, Grover's search \cite{grover1996}, linear inversion \cite{harrow2009} and Hamiltonian simulation/quantum chemistry \cite{berry2015a}, which will out-perform the best-known subroutines on a classical computer. We can produce small-scale devices that promise quantum advantage \cite{arute2019,zhong2020,madsen2022}, and error correction is reaching the break-even point where it should be scalable \cite{sivak2023}. This scaling up of error correction is the most significant barrier to achieving a universal fault-tolerant quantum computer that will implement all these amazing algorithms.

Quantum error correction, however, is not only an engineering problem. It is a challenge to theorists to create high performance codes. An $[[n,k,d]]$ code is specified by 3 key properties: $n$, the number of physical qubits that encode the $k$ logical qubits, and the distance $d$, which indicates the number of errors that can be tolerated. Small-scale codes work remarkably well. For example, the 5-qubit code \cite{laflamme1996} utilises the state space so well, that it is said to be perfect. Steane's 7-qubit code \cite{steane1996a} similarly is perfect within the restrictions of CSS codes, while having a near-optimal set of transversal logical gates. However, just as with classical error correction, where scaling up the repetition code to large systems fails to pack data as densely as we otherwise might, encoding many logical qubits as densely as possible is far from trivial.

We can assess the storage density by evaluating $kd^2$. Since neither can be larger than $\order{n}$, $kd^2\sim n^3$ at best and any code saturating this is said to be `good'. Contrast this with the performance of many copies of the Steane code: an $[[n,n/7,3]]$ code with $kd^2\sim n$; far from good, but optimal performance for $k$ is easy to achieve.

The most commonly considered methods for scaling an error correcting code to large systems are concatenated codes \cite{aliferis2006} and surface codes, such as the Toric code \cite{kitaev2003}. The $[[2d^2,2,d]]$ Toric code achieves large distance, which gives us the ability to adapt to a broad range of scenarios, but only has $kd^2\sim n$. Meanwhile, a concatenated Steane code with $D$ levels of concatenation, is a single code $[[7^D,1,3^D]]$. Again, we can control the distance, but only achieve $kd^2\sim n^{2\log_73}\approx n^{1.13}$. It should be noted that there are more sophisticated ways to perform concatenation which achieve much better performance \cite{yamasaki2024}.

One major advantage of the Toric code is its formulation -- the syndrome extraction interrogates small clusters of qubits that are neighbours on some two-dimensional lattice, complying with the constraints of many physical implementations of quantum computers. Indeed, under these locality restrictions, the code is asymptotically optimal \cite{bravyi2010b}; it is impossible for a code that is local in two dimensions to exceed $kd^2\sim n$. Our search for well-performing codes must therefore violate the locality restriction. There are certainly experimental scenarios that can take advantage of this lifting of restrictions \cite{bluvstein2024}.

In studying $[[n,k,d]]$ codes without the locality restriction, we still choose to retain the property that the check operators should only act on a small number of qubits. We say the code is ``low density'', and hence the codes are known as low density parity check codes (LDPC). Classical versions of these codes are now extensively used in telecommunications. In the quantum regime, the best-known codes were, for a long time, the hypergraph product codes \cite{tillich2014}, limited to the performance of
\begin{equation}
kd^2\sim n^2.
\label{eq:hypergraph_limit}
\end{equation}
Clearly, an exciting improvement, but still a large step away from being good codes, with the limitation being achieving large enough $d$. Then, in about 2022, a flurry of results \cite{panteleev2022,panteleev2022a,evra2022} first showed that this performance could be exceeded, and ultimately resulted in good codes being produced. The tipping point was an initial example due to Hastings, Haah \& O'Donnell \cite{hastings2021} that showed \cref{eq:hypergraph_limit} could be violated, and its generalisation by Breuckmann \& Eberhardt \cite{breuckmann2021}, achieving the same
$kd^2\sim n^{2}$
but crucially realising it with $d\sim n^{3/5}$, where previous studies were limited to $d\sim\sqrt{n}$.


Breuckmann \& Eberhardt \cite{breuckmann2021} take the well-characterised LPS family of expander graphs \cite{lubotzky1988} and formulate a `balanced product' with a repetition code of appropriate size.
However, this product was specified in an intensely mathematical way which is intimidating to mere mortals, such as ourselves. In this paper, we express the constructions in the language of parity check matrices and provide insight into the key results, making the conversion from graph to code particularly straightforward, \cref{eq:balanced}. Not all proofs are complete -- we occasionally favour pedagogy, knowing that the rigorous results have appeared elsewhere \cite{breuckmann2021}.

\section{Quantum Error Correcting Codes}

In this section, we will briefly review some of the basics of error correction \cite{roffe2019}. It is assumed that the interested reader will already be familiar with this, and our primary focus is to establish notation.

In classical error correction, we can specify an error correcting code by its parity-check matrix $H\in\{0,1\}^{(n-k)\times n}$. A binary string $x\in\{0,1\}^n$ is a codeword if $Hx=0$, with the idea being that if any pair of codewords differ in at least $d$ places, evaluating non-trivial $Hy$ allows us to detect the presence of some errors and, under the assumption that fewer than $d/2$ errors have occurred, to return $y=e\oplus x$ to the original $x$ by determining $e=\underset{v:H(y\oplus v)=0}{\text{argmin}} |v|$. The code is said to be low-density if every row and every column of $H$ has a weight bounded by some constant, $s$.

In a quantum setting, we can use the same classical code to protect against bit-flip errors $X_i$. We start with a codeword as a quantum state, $\ket{x}$. It might acquire some bit-flip errors $X_e$ such that the state becomes $\ket{y}=\ket{e\oplus x}=X_e\ket{x}$. Here, we use the notation
$$
X_e=\bigotimes_{i=1}^nX^{e_i}.
$$
to express the idea that Pauli $X$ operations (bit-flips) have been applied to sites $i$ where $e_i=1$. For each parity-check of the code, $r\in H$, we associate a measurement
$$
Z_r=\bigotimes_{i=1}^nZ^{r_i}.
$$
This gives a measurement outcome $\bra{y}Z_r\ket{y}=(-1)^{r\cdot y}=(-1)^{r\cdot e}$, exactly the same as the result of computing the parity check in the classical case.

In a quantum code, we must protect both against $X$ errors and $Z$ errors. CSS codes \cite{calderbank1996,steane1996a}, to which we limit ourselves, separate the challenge of error detection, recognising that if one can independently detect and correct $X$ and $Z$ errors on a single qubit, then one can correct any single-qubit error. It divides that task into a parity-check for $X$ errors, made up of Pauli-$Z$ strings specified by a parity-check matrix $H_Z$ (as above), and a parity-check for $Z$ errors, which uses Pauli $X$-strings specified by $H_X$. Specifically, on $n$ physical qubits we might have $m_x$ linearly independent parity checks \footnote{Beware: in what follows, the rows are not linearly independent.} in $H_X\in\{0,1\}^{m_X\times n}$ such that each row $r\in H_X$ defines a single stabilizer $X_r$.
Similarly, for $H_Z\in\{0,1\}^{m_Z\times n}$, a row $r$ defines an operator $Z_r$. In order to be able to measure all of these values simultaneously (and not have the results of one change the system so that previous measurements are irrelevant), all the measurement operators must commute. All $X$-types commute with each other, but we require
$$
[X_r,Z_{r'}]=0\qquad\forall r\in H_X, r'\in H_Z.
$$
This is succinctly summarised by the property
\begin{equation}
H_X\cdot H_Z^T=0, \label{eq:anticommute}
\end{equation}
where all arithmetic is performed modulo 2. 
Since the operators all commute and square to $\identity$, they are stabilizers. Indeed, the codewords of a quantum error correcting code are those states $\ket{\psi}$ which are stabilized by the parity-checks:
$$
Z_r\ket{\psi}=\ket{\psi}\quad\forall r\in H_Z,\qquad X_r\ket{\psi}=\ket{\psi}\quad\forall r\in H_X.
$$
The dimension $k=n-m_A-m_B$ of the space spanned by these states determines the number of logical qubits. For each logical qubit $i$, we associate logical $X_{L,i}$ and logical $Z_{L,i}$ operators, themselves comprising Pauli $X$ and $Z$ strings respectively. They commute with all the stabilizers and generate the algebra of Pauli operators, 
$$
[X_{L,i},X_{L,j}]=[Z_{L,i},Z_{L,j}]=0
$$
and, especially,
$$
\{X_{L,i}, Z_{L,i}\}=0,\qquad [X_{L,i},Z_{L,j}]=0
$$
where the latter only holds for $i\neq j$.

Logical operators are not unique. We can `dress' them with stabilizers, so, for example, $X_{L,1}$ and $X_{L,1}X_r$ for $r\in H_X$ have the same effect, preserving the logical space and the anti-commutation relations. This gives us a lot of flexibility in their description, and makes the problem of determining the distance of a code (the smallest possible weight of a logical operator) highly non-trivial.

The key properties of a code are summarised by the $[[n,k,d]]$ notation, conveying that there are $n$ physical qubits, $k$ logical qubits, and the code has a distance $d$. We will use two variants of this notation. In the case where the distance is asymmetric, meaning that we can tolerate a different number of $X$ errors compared to $Z$ errors, we write $[[n,k,d_X,d_Z]]$. When asymptotics are the focus, we might denote $\order{[[n,n^{3/4},n^{1/2}]]}$, for example, to convey that when there are $n$ physical qubits, the number of logical qubits is $\order{n^{3/4}}$ and the distance is $\order{n^{1/2}}$. Throughout this paper, we have used a SAT solver \cite{kay2025} to find the code distance of any examples, such as those found in \cite{kay2025a}. The advantage of the SAT solver is that it does not have to exhaustively search every possible error string to find the shortest logical operator; it can skip parts of the search space and provides a checkable proof that there isn't a smaller logical operator.

\subsection{The LDPC Challenge}

In classical error correction, it's relatively easy to make a good low density parity check (LDPC) code \cite{gallager1963}. Pick any sufficiently large random binary matrix with appropriate row and column weights, and it is almost certain that the number of logical bits and distance against errors both scale linearly in the number of physical qubits! Why, then, is it so challenging to build quantum codes?

Imagine that we started with a good classical LDPC code's parity check matrix as the check matrix for $X$ errors ($H_Z$). By definition, \emph{all} low weight $X$-type operators yield a non-trivial syndrome. So there are no low-weight $X$-type operators ($H_X$) available for detecting the presence of $Z$ errors as they must fulfil \cref{eq:anticommute}. It requires some difficult balancing in order to be able to have both low weight $X$-type and $Z$-type stabilizers while leaving a large distance for a non-trivial number of logical qubits.
 One elegant design is the hypergraph product code \cite{tillich2014}
\begin{align*}
H_X&=\begin{bmatrix}
H_1\otimes \identity & \identity\otimes H_2
\end{bmatrix} \\
H_Z&=\begin{bmatrix}
\identity\otimes H_2^T & H_1^T\otimes\identity
\end{bmatrix}
\end{align*}
with which it is particularly straightforward to verify \cref{eq:anticommute}. We are free to pick $H_1$ and $H_2$, but the tensor product imposes a blow-up in the number of qubits that the distance will never be able to compensate for. Chain complexes have proved a particularly insightful route for yielding new structures of parity-check matrix that also obey \cref{eq:anticommute}, and we will make extensive use of two such structures in this review, while studiously avoiding further mention of chain complexes.

\section{Balanced Product}

The main construction of \cite{breuckmann2021} is the balanced product code. While very general, restrictions were made for the ensuing analysis. Specifically, in the product of two classical codes, one was taken to be a repetition code of a defined size. We make those assumptions from the beginning. A more general case is discussed in \cref{sec:missing_product}, which encompasses both this case and the hypergraph product as opposite extremes. Hence, we give ourselves the maximal opportunity to observe behaviour that diverges from the hypergraph product.

\begin{definition}\label{def:balanced_product}
Consider a binary matrix $I$ which has an order $l$ permutation symmetry, meaning there are permutation operations $R^TR=\identity, \ C^TC=\identity$ such that
$$
RI=IC^T,\qquad R^l=\identity,\qquad C^l=\identity,
$$
and all the orbits of $R,C$ are the same length, $l$.

\noindent The paritycheck matrices of the balanced product are:
\begin{equation}\label{eq:balanced}
H_X=\begin{bmatrix}
I^T & \identity+C
\end{bmatrix}, \qquad
H_Z=\begin{bmatrix}
\identity+R & I
\end{bmatrix}.
\end{equation}
\end{definition}
\cref{eq:balanced} clearly satisfies \cref{eq:anticommute}, and permits free choice of the binary matrix $I$ \footnote{Indeed, this would still be true if we replaced $\identity+C$ and $\identity+R$ with $f(C)$ and $f(R)$ respectively, where $f(x)$ is a low weight polynomial with binary coefficients.}. If $I$ is LDPC, then so is the new code. The lack of the tensor product substantially reduces the number of physical qubits involved. Note that the matrices $\identity+C$ and $\identity+R$, taken alone, would define order $l$ repetition codes. 

\begin{example}\label{ex:repetition_small}
Let $I$ be the incidence matrix of a 6-cycle,
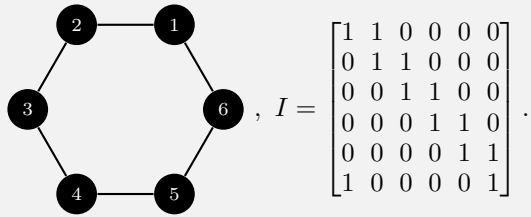
\begin{center}
\begin{tabular}{cc}
\begin{tikzpicture}[baseline={([yshift=-axis_height]A6)}]
\foreach \x in {1,2,3,4,5,6} {
\node [circle,fill=black,minimum width=0.5cm,shift=({60*\x}:1.3),text=white] (A\x) at (0,0) {$\scriptstyle\x$};
}
\draw [thick] (A6) -- (A1)--(A2)--(A3)--(A4)--(A5)--(A6);
\end{tikzpicture}
&
$
,\ I=\begin{bmatrix} 1 & 1 & 0 & 0 & 0 & 0 \\ 0 & 1 & 1 & 0 & 0 & 0 \\ 0 & 0 & 1 & 1 & 0 & 0 \\ 0 & 0 & 0 & 1 & 1 & 0 \\ 0 & 0 & 0 & 0 & 1 & 1 \\ 1 & 0 & 0 & 0 & 0 & 1 \end{bmatrix}.
$
\end{tabular}
\end{center}

Let $P$ by the cyclic permutation of 6 vertices, so $I=\identity+P$. $I$ has many symmetries; we choose an order 3 symmetry $R=C^T=P^2$. The orbits in this case are $(1,3,5)$ and $(2,4,6)$. We proceed to construct the balanced product
\begin{align*}
H_X&=\begin{bmatrix} \identity+P^T & \identity+{P^T}^2 \end{bmatrix} \\
H_Z&=\begin{bmatrix} \identity+P^2 & \identity+P\end{bmatrix}.
\end{align*}
The explicit matrices can be found in \cite{kay2025a}. 
One can verify that $H_X\cdot H_Z^T=0$ (modulo 2), and that there are logical operators
\begin{align*}
X_{L,1}&=X_1X_3X_5,\\
Z_{L,1}&=Z_5Z_6Z_8Z_9,\\
X_{L,2}&=X_2X_4X_{11}X_{12},\\
Z_{L,2}&=Z_7Z_9Z_{11}.
\end{align*}

This CSS code is $[[12,2,3]]$. Its performance already exceeds the $[[14,2,3]]$ of two copies of the CSS-perfect Steane code.

\end{example}

\subsection{Repetition Code}
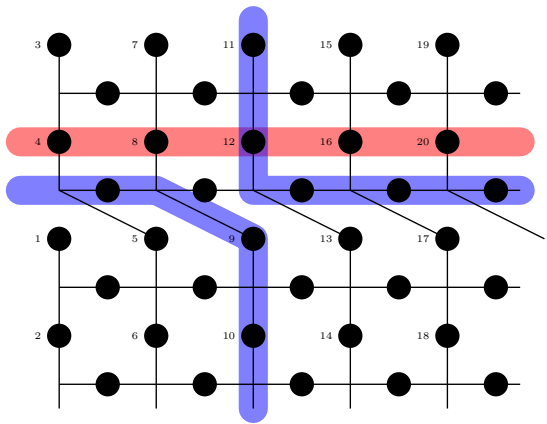
\begin{figure}
    \centering
    \begin{adjustbox}{max width=0.4\textwidth}
\begin{tikzpicture}
\foreach \x[evaluate={\xi=ifthenelse(\x==5,2,2*\x+2)}] in {1,...,5} {%
  \foreach \y[evaluate={\yi=ifthenelse(\y<3,ifthenelse(int(4*(\x) + \y)==22,20,int(Mod(4*(\x) + \y -2,20))),int(Mod(4*(\x-1) + \y -2,20)));}] in {1,...,4} {%
    \node [circle,fill=black, minimum width=0.5cm,label={left:{$\scriptstyle \yi$}}] (m-\x-\y) at (2*\x,-2*\y) {}; 
    \node [circle,fill=black, minimum width=0.5cm] at (2*\x+1,-2*\y-1) {}; 
    }
  \draw [thick,black] (2*\x,-2) -- (2*\x,-5) -- (2*\x+2,-6) (\xi,-6) -- (\xi,-9.5);
    }
\foreach \y in {1,...,4} {%
  \draw [thick,black] (2,-2*\y-1) -- (11.5,-2*\y-1);
}
\begin{scope}[on background layer]
\draw[line width=0.6cm,red,opacity=0.5,line cap=round] (1.2,-4) -- (11.5,-4);
\draw[line width=0.6cm,blue,opacity=0.5,line cap=round, line join=round] (6,-1.5) -- (6,-5) -- (11.5,-5) (1.2,-5) -- (4,-5) -- (6,-6) -- (6,-9.5);
\end{scope}
\end{tikzpicture}
\end{adjustbox}
\caption{Depiction of code based on length 20 repetition code, using length $l=5$ cycle for symmetry. Qubits in left-hand block are numbered consecutively. $X$-type stabilizers act on the 4 qubits that neighbour a vertex. $Z$-type stabilizers act on faces. Boundary conditions are periodic. A pair of anti-commuting logical operators are shown.}\label{fig:toric_like}
\end{figure}
In \cref{ex:repetition_small}, we examined a simple case where $I$ is a repetition code of size 6. In general, consider a repetition code of size $lm$ and take an additional symmetry of orbit length $l$, so there are $m$ orbits. If we lay the qubits out on two copies of an $l\times m$ grid, then the code looks almost like the Toric code -- each stabilizer acts on 4 qubits. The only difference is the behaviour of the boundary conditions across one boundary; the balanced product has introduced a skew (see \cref{fig:toric_like}), and it is this skew that has the potential to increase the code distance.

If we wanted to construct logical operators, our first attempt would be to write down
\begin{align*}
X_{L,1}&=\prod_{i=0}^{l-1}X_{im+1},& \qquad& X_{L,2}=\prod_{i=1}^{lm}X_{lm+i}, \\
Z_{L,1}&=\prod_{i=1}^{lm}Z_i,&\qquad& Z_{L,2}=\prod_{i=0}^{l-1}Z_{2lm-im+1}.
\end{align*}
Each logical qubit is defined with one logical Pauli supported on every site of either the left- or right-hand block. The other just loops around one orbit. Provided $l$ is odd, the two anti-commute, so let's assume $l$ is odd, $m$ is even. The single-orbit terms are just the same as the horizontal (say) logical operators of the Toric code, and the weight $l$ is their corresponding distance; one cannot make them shorter by dressing with stabilizers. The whole-block operators are the equivalent of the vertical operators, except that the skew means that the ends don't match after a single orbit. However, in this case, these operators can be dressed by stabilizers to decrease their weight. For example, with $Z_{L,1}$, we can take the stabilizer that acts on $Z_1Z_{m+1}Z_{lm+1}Z_{lm+2}$. Dressing with this preserves the weight. However, if we also dress with $Z_2Z_{m+2}Z_{lm+2}Z_{lm+3}$, the net weight drops by 2. We can perform a similar trick for each orbit, creating a run of terms covering all but one element of the orbit (since $l$ is odd) on the left-hand block, and replacing them with just $Z_{lm+1}Z_{lm+m}$ on the right-hand block. We thus arrive at a dressed logical operator with a weight $m+l-1$, as depicted in \cref{fig:toric_like}. We claim this is the minimum weight of the logical operator. While the code is strictly still $[[lm,2,l]]$, if we choose a subsystem code in which we just use one of the logical qubits, we can instead treat it as an $[[lm,1,l+m-1,l]]$ code, which therefore has a potential advantage due to its asymmetric distance. Note that the code defined in this way is still local, and therefore still constrained by the same bound \cite{bravyi2010b}.

Of course, while the repetition code yields the optimal distance for the storage of a single bit, the magic of coding since Shannon's Noisy Coding Theorem is that error correcting codes don't have to sacrifice much distance to gain a lot more storage capacity. We should therefore expect the possibility to vastly outperform the Toric code through judicious choice of $I$.

\subsection{Expanders}

Key to ensuring a large code distance in the balanced product is the imposition that $I$ should be \emph{expanding}:
\begin{definition}
A matrix $I\in\{0,1\}^{m\times n}$ is said to be $(\alpha,\beta)$ expanding if:
\begin{itemize}
\item $\forall x\in\{0,1\}^n:|x|\leq \alpha n,\ |Ix|\geq\beta|x|$
\item $\forall y\in\{0,1\}^m:|y|\leq \alpha m,\ |I^Ty|\geq\beta|y|$.
\end{itemize}
\end{definition}
\noindent For example, if $\beta>1$ then for all sufficiently short strings $x$, $Ix$ has greater weight than $x$; it has been expanded. The incidence matrix or biadjacency matrix of an expander graph is an expanding matrix \cite{panteleev2022a}.

Indeed, we can now state the main theorem.

\begin{theorem}\label{thm:main}
Let $I\in\{0,1\}^{m\times n}$ be an $(\alpha,\beta)$-expanding matrix which obeys the properties of \cref{def:balanced_product}, has row and column weights of at most $s$, and
$$
\max(|\text{ker}(I)|,|\text{ker}(I^T)|)=k_0
$$ 
The balanced product code of \cref{def:balanced_product} has properties
$$
[[n+m,\frac{k_0}{l},d_X,d_Z]].
$$
The code is LDPC. There exist positive constants $\gamma_X$ and $\gamma_Z$ such that $d_X=\gamma_Xl$ \footnote{assuming a large enough code $2\beta\gamma_Xl<\alpha n$ that a token case can be neglected.} and $d_Z=\gamma_Zm$.
\end{theorem}

Note that we will typically be working in a situation where $m>n$, which causes us to focus on the left-hand block (since $I^T$ has more columns than rows, and therefore has a null space that we can make use of). If it were the case that $m<n$ (as in section \ref{sec: Bringing it Together}), we would instead focus on logical qubits localised on the right-hand block. The analysis would be essentially unchanged.

\begin{proof}
The number of physical qubits follows a straightforward examination of the matrix sizes, $R\in\{0,1\}^{m\times m}$, $C\in\{0,1\}^{n\times n}$ and $I\in\{0,1\}^{m\times n}$.

The total number of rows of $H_X$ and $H_Z$ is $n+m$, the same as the number of columns. Hence, the logical qubits derive from the linear dependence of the rows.
$$
k=|\text{ker}(H_X^T)|+|\text{ker}(H_Z^T)|
$$
To find a null vector of $H_Z^T$, we seek a $v$ such that $vH_Z=0$, requiring both $(\identity+C^T)v=0$ and $Iv=0$. Thus,
$$
|\text{ker}(H_Z^T)|=|\text{ker}(\identity+C^T)\cap\text{ker}(I)|.
$$
However, we know essentially nothing about $|\text{ker}(I)|$ \footnote{Commonly, $I$ has a fixed row weight, so it has a null vector comprising the all-ones. This is often the only one.}, so we'll simply take the bound $|\text{ker}(H_Z^T)|\geq 0$. For $H_X$, we are given that $|\text{ker}(I^T)|=k_0$. However, we need
$$
|\text{ker}(H_X^T)|=|\text{ker}(\identity+R^T)\cap\text{ker}(I^T)|.
$$
The null vectors $v$ of $I^T$ ($vI=0$) might not have the symmetry of $R$. If not, however, then $vR$ must also be a null vector, and $vR^2$, etc. At worst, it would take $v\sum_{i=0}^{l-1}R^i$ to create a null vector of $R$ as well. Hence,
$$
k\geq|\text{ker}(H_X^T)|\geq \frac{k_0}{l}.
$$
This is our lower bound on the number of logical qubits.

We defer a full distance proof until \cref{sec:gory}, as this is the key technical challenge. Instead, we will motivate the result. Imagine that the lowest weight representation of the operators coincided with their being entirely localised either on the left-hand or right-hand block only. If so, then the left-hand logical qubits all have logical $Z$ operations of the form
$$
l_Z=\begin{bmatrix} u \\ 0 \end{bmatrix}.
$$
They must satisfy $H_Xl_Z=0$, and thus $I^Tu=0$. However, recall that $I$ is expanding -- if $|u|\leq\alpha m$ then $|I^Tu|\geq\beta |u|>0$. In other words it must be that $|u|>\alpha m$. Its left-block $Z$ distance scales linearly in the system size. It is quite obvious how to formulate the matching logical $X$ -- simply find any unit vector $e$ and create an all-ones vector over an orbit of $R$: $\sum_{i=0}^{l-1}R^ie$. These have weight $l$.

We can give a similar argument for the right-hand logical operators, although the asymmetry in the knowledge of null vectors of $I$ and $I^T$ is substantial enough that, in fact, we will focus entirely on the left-hand logical qubits, fixing all the right-hand qubits to be in their $\ket{0}$ state; we will be using a subsystem code.

\end{proof}

The details presented so far are actually sufficient to start explicitly constructing logical operators. If $v$ is a null vector of $I^T$, $I^Tv=0$ then we can use
$$
l_Z=\begin{bmatrix}
\sum_{i=0}^{l-1}R^iv \\ 0
\end{bmatrix}.
$$
We have already argued that there are at least $k_0/l$ of these. As ever, we must worry whether these are truly logical operators, or just a product of stabilizers. We prove that they are distinct logical operators (such that the counting covers all $k_0/l$ logical operators) by finding unique anti-commutation relations. To that end, think of the matrix of logical $Z$ operators $L_Z$. There must be individual columns where, for each logical operator, it is the only one that acts on that site \footnote{From a programming perspective, apply row reduction modulo 2. Those vertices are identified by the $\identity$ block.}. Use these as the $e$ for the corresponding $l_X$:
$$
l_X=\begin{bmatrix}
\sum_{i=0}^{l-1}R^ie \\ 0
\end{bmatrix}.
$$
We are guaranteed that $l_X\cdot l_Z=l$, so assuming odd $l$, these are definitely logical operators yielding a qubit. Moreover, by design,  for any $l_X' \neq l_X$, we have $l_X'\cdot l_Z=0$; they must define different logical qubits.


\subsection{Distance Proof -- The Full and Gory Details}\label{sec:gory}

To move beyond the motivating arguments of the previous section that attempt to justify the code distance, we will present the full calculation, following \cite{panteleev2022a}. For the left-hand $Z$ operators we have that when they are written in the form
$$
\begin{bmatrix} u' \\ 0 \end{bmatrix},
$$
they have weight $>\alpha m$. But there are many other ways in which they could be written, using a dressed version of the logical operators,
$$
\begin{bmatrix} u \\ v \end{bmatrix}.
$$
We need to prove a lower bound on $|u|+|v|$ for all possible ways of writing the same logical operator.
As stabilizers, $$I^Tu'=0=I^Tu+(\identity+C)v,$$ where there exists an $h$ such that
$$
\begin{bmatrix} u & v\end{bmatrix}=\begin{bmatrix} u' & 0 \end{bmatrix} + h\cdot H_Z
$$
It must be that $v=I^Th$ and $u=u'+(\identity+R^T)h$. We can verify
\begin{align*}
I^Tu+(\identity+C)v&=I^T(u'+(\identity+R^T)h)+(\identity+C)I^Th \\
&=I^Tu'+(\identity+C)I^Th+(\identity+C)I^Th \\
&=0,
\end{align*}
making use of the symmetry of $I$.

Why is this important? Our aim is to prove that either $|u|$ or $|v|>\gamma_Z m$. Let's work in the regime $|u|<\gamma_Z m$, otherwise we're trivially done. We will define
$$
u^{(t)}=\left(\sum_{i=0}^{t-1}{R^T}^i\right)u.
$$
This means that $(\identity+R^T)u^{(t)}=(\identity+{R^T}^t)u$. Hence,
\begin{align*}
u^{(l)}&=\left(\sum_{i=0}^{l-1}{R^T}^i\right)u \\
&={u'}^{(l)}+(\identity+{R^T}^l)h \\
&={u'}^{(l)},
\end{align*}
which must be a null vector on just the left-hand side with weight at least $\alpha m$. Meanwhile, $u^{(1)}=u$, such that $|u^{(1)}|<\gamma_Z m$. Using the freedom to pick $\gamma_Z$, if we select $\gamma_Z<\alpha$, there must exist a $t_0$ which is the smallest $t$ such that $|u^{(t_0+1)}|>\alpha m$. In this case, $u^{(t_0+1)}=u^{(t_0)}+{R^T}^{t_0}u$, so we can give a weight bound of
$$
|u^{(t_0)}|=|u^{(t_0+1)}+{R^T}^{t_0}u|\geq(\alpha-\gamma_Z)m.
$$
Now,
\begin{align*}
(\identity+C^{t_0})v&=\left(\sum_{i=0}^{t_0-1}C^i\right)(\identity+C)v \\
&=\left(\sum_{i=0}^{t_0-1}C^i\right)I^Tu \\
&=I^Tu^{(t_0)} \\
2|v|&\geq \beta(\alpha-\gamma_Z)m,
\end{align*}
where we have again invoked expansion because $|u^{(t_0)}|<\alpha m$. Provided we select $\gamma_Z$ satisfying the already specified $\gamma_Z<\alpha$ and the new $\gamma_Z<\beta(\alpha-\gamma_Z)/2$, this works. Thus, we can select anything smaller than
$$
\gamma_Z=\alpha\min\left(1,\frac{\beta}{2+\beta}\right),
$$
thus bounding the $Z$ distance by $\gamma_Z m$.

The $X$ distance can follow in a similar manner. Let us now assume that both $|u|,|v|<\gamma_X l$. We will prove a contradiction with the expansion properties of $I$. This time, we have a left-localised vector
$$
\begin{bmatrix} u' \\ 0 \end{bmatrix}
$$
which must satisfy $u'=Ru'$, and a dressed one
$$
\begin{bmatrix} u \\ v \end{bmatrix}=\begin{bmatrix} u'+Ih \\ (\identity+C^T)h \end{bmatrix}.
$$
We shall also define $v^{(t)}=(\identity+{C^T}^{t})h$. In particular, note that $v^{(1)}=v$. Now,
\begin{align*}
|Iv^{(t)}|&=|(\identity+R^t)Ih| \\
&=|(\identity+R^t)(u+u')| \\
&=|(\identity+R^t)u| \\
&\leq 2|u| \\
&\leq 2\gamma_X l.
\end{align*}
By way of contrast, let $t_0$ be the smallest integer such that $|v^{(t_0)}|>\frac{1-\alpha}{s} l$ (we will prove that such a number exists in \cref{lem:fiddly}). Then we have
\begin{align*}
|v^{(t_0)}|&=|(\identity+C^T)h+C^T(\identity+{C^T}^{t_0-1})h| \\
&\leq |v| + |v^{(t_0-1)}| \\
&\leq \alpha l+\frac{1-\alpha}{s} l \\
&=\frac{(s-1)\alpha+1}{s}l.
\end{align*}
So, if $\frac{(s-1)\alpha+1}{s}l<\alpha n$, $v^{(t_0)}$ should obey expansion properties, i.e.\ $|Iv^{(t_0)}|>\beta|v^{(t_0)}|$.
However, if we choose $\gamma_X<\frac{\beta (s-1)\alpha+1}{2 s}$ then we see that $|Iv^{(t_0)}| \leq \beta \frac{(s-1)\alpha + 1}{s}l$, which contradicts the expansion properties.

\begin{lemma}\label{lem:fiddly}
For any $h$, there exists a choice of $t$ such that
$$
|v^{(t)}|>\frac{1-\alpha}{s}l
$$
given that $|u|<\alpha l$.
\end{lemma}
\begin{proof}
Recall that $h$ is selecting a linear combination of rows of $H_X$. However, we built certain symmetry properties into $H_X$. We thus consider decomposing $h=\sum_ih_i$ where each $h_i$ only has support on a single orbit of $C^T$. The intersections of different $h_i$ are trivial: $h_i\cdot h_j=0$ of $i\neq j$. Furthermore, it is always sufficient to select $|h_i|\leq l/2$ because, on the right-hand block, $h_i$ being all-ones yields an all-zeros result, i.e.\ we can arbitrarily switch between $h_i$ and $\bar h_i$. We just happen to always choose the one of lower weight.

We will now show that the average value of $|v^{(t)}|$ is at least $|h|$, meaning that at least one value of $t$ exists such that $|v^{(t)}|\geq |h|$. We have
\begin{align*}
|v^{(t)}|=|h+{C^T}^th|
=2|h|-2h\cdot({C^T}^th)
\end{align*}
The average is then given by
\begin{align*}
\sum_t|v^{(t)}|&=2l|h|-2\sum_th\cdot({C^T}^th) \\
&=2l|h|-2\sum_t\sum_ih_i\cdot({C^T}^th_i).
\end{align*}
For every pair of bits $j,k$, there is a choice of $t=k-j$ that causes them to overlap. There are $|h_i|^2$ such pairs, so the average is
\begin{align*}
\sum_t|v^{(t)}|&=2l|h|-2\sum_i|h_i|^2 \\
&\geq 2l|h|-2\frac{l}{2}\sum_i|h_i| \\
&=l|h|.
\end{align*}
Finally, we need a bound on $|h|$. Since the columns of $I$ contain no more than $s$ non-zero entries,
\begin{align*}
|h|&\geq\frac{|Ih|}{s} \\
&=\frac{1}{s}|u+u'| \\
&\geq\frac{1}{s}(|u'|-|u|).
\end{align*}
When $|u|<\alpha l$ and $|u'|\geq l$, this gives that
$$
|h|\geq\frac{1-\alpha}{s}l.
$$
\end{proof}

This completes the proof of the distance properties of the balanced product code claimed in \cref{thm:main}. In particular, for the left-hand block the $Z$-distance is linear in $m$, while the $X$ distance is linear in the length of the orbit of the symmetries, $l$. Similarly, for the right-hand block, the $X$ distance in linear in $n$ and the $Z$ distance in linear in $l$. We always focus on just one of these halves by using a subsystem code, concentrating on the one with the largest null space for $I$ or $I^T$. 

\section{Distance Balancing}

The limitation of the balanced product construction is that the distances against $X$ and $Z$ errors may be different. We aim to redress that balance, based on \cite{evra2022}; a process known as distance balancing.

\begin{definition}\label{def:distance_balance}
Given a quantum code comprising parity check matrices $H_X$ and $H_Z$ which yield distances $d_Z$ and $d_X$ respectively, these distances can be balanced using the parity-check matrix $H_C$ of a good classical LDPC code.
\begin{equation}
\begin{aligned}
\tilde H_Z&=\begin{bmatrix}
H_Z\otimes\identity & 0 \\
\identity\otimes H_C & H_X^T\otimes\identity
\end{bmatrix}\\
\tilde H_X&=\begin{bmatrix}
H_X\otimes\identity & \identity\otimes H_C^T
\end{bmatrix}
\end{aligned}
\end{equation}
\end{definition}

We have set this up assuming $d_X<d_Z$, but this can be switched around if the imbalance in the distance is in the opposite direction.

The number of physical qubits involved in the new code is immediate: $N=nn_C+m_X m_C$. If the quantum code encoded $k$ logical qubits, and the classical code encodes $k_C$ logical bits, we claim that the distance balanced code yields
$
K=kk_C
$
logical qubits.

To see this, let's assume that we have removed all linearly dependent rows from $H_X$, $H_Z$, and all the linearly dependent columns from $H_C$. The remaining sizes are $H_X\in\{0,1\}^{m_X\times n}$, $H_Z\in\{0,1\}^{m_Z\times n}$ and $H_C \in\{0,1\}^{m_C\times n_C}$. As a result, we have $k=n-m_X-m_Z$ and $k_C=n_C-m_C$. While all the rows of $\tilde H_X$ are linearly independent, the rows of $\tilde H_Z$ may not be. Let's denote by $r$ the dimension of the linear dependence. Then we have
\begin{align*}
K&=(nn_C+m_Xm_C)-m_X n_C-(m_Zn_C+nm_C)+r \\
&=kk_C-m_Zm_C+r.
\end{align*}
We just have to prove that $r=m_Zm_C$. We are searching for bit strings of the form
\begin{align*}
0&=\begin{bmatrix}
a\otimes b & c\otimes d
\end{bmatrix}\tilde H_Z \\
& = \begin{bmatrix} aH_Z\otimes b+c\otimes d H_C  & cH_X^T \otimes d \end{bmatrix}.
\end{align*}
If it were the case that $d=0$, then we'd need $aH_Z=0$, which has no solutions by construction. Thus, the only way we can solve this is with $H_Xc=0$, $c=H_Z^Ta$ and $b=H_C^Td$. (For any $a$, $H_Xc=H_XH_Z^Ta=0$ by \cref{eq:anticommute}). Hence, this is a solution for all $a\in\{0,1\}^{m_Z}$ and $d\in\{0,1\}^{m_C}$. Thus, as required, $r=m_Zm_C$.

The next step, of course, is to argue the distances. We want to see that these are given by
\begin{align*}
\tilde d_X&=d_Xd_C \\
\tilde d_Z&=d_Z
\end{align*}
where the matrices $H_X,H_Z,H_C$ have distances $d_Z,d_X,d_C$ respectively. The plan would then be to select $d_C=d_X/d_Z$, such that the distances against the two types of error will be the same. Notice that we can construct logical operators of the form
$$
l_X=\begin{bmatrix}L_X\otimes L \\0\end{bmatrix}, \qquad l_Z=\begin{bmatrix} L_Z\otimes e \\ 0 \end{bmatrix},
$$
where $L_X$ represents a logical $X$ of the original code, $L$ is a codeword of the classical code, and $e$ is a unit vector such that $e\cdot L=1$. This clearly gives logical operators of weight $d_Xd_C$ and $d_Z$ respectively. But could there be anything shorter? We will not give the full details here. Much as before, we have to ensure that there are no shorter dressed operators, or at least that dressing the logicals does not fundamentally change the scaling.

In general, this process allows us to transform an $[[n,k,d_X,d_Z]]$ code into an $\order{[[n\frac{\max(d_X,d_Z)}{\min(d_X,d_Z)},k\frac{\max(d_X,d_Z)}{\min(d_X,d_Z)},\max(d_X,d_Z)]]}$ code.

\section{Bringing it Together}\label{sec: Bringing it Together}

The procedure is now clear: find a classical expander $I$ with an appropriate symmetry, build the distance asymmetric code specified in \cref{def:balanced_product}, and then rebalance the distance by building it into \cref{def:distance_balance}.

The challenge, then, is to find a family of expanders $I$ with the right symmetry properties so that we can understand the behaviour of the codes as we change the number of physical qubits involved. Few families of classical LDPC codes are known.

One particularly simple construction that replicates the performance of the hypergraph product codes is to pick a random expander graph with incidence matrix $I_0\in\{0,1\}^{m\times n}$. Then set $I=I_0\otimes\identity_l$. If $M$ is a cyclic permutation over $l$ vertices, then we can define
$$
R=\identity_m\otimes M,\qquad C=\identity_n\otimes M^T.
$$
With high probability, we get a code that is $\order{[[nl,nl,l,n]]}$ before distance balancing. In fact, distance balancing is not even necessary here simply by selecting $l\sim n$. Then the code is $\order{[[N,N,\sqrt{N}]]}$. Instead, with $I_0=\identity+M$ (not an expander), we recover the familiar Toric code. Of course, the tensor product structure misses out on what could be much denser packings of information in the LDPC codes. 

\subsection{LPS Expander Graphs}

The family of expander graphs introduced by Lubotzky, Phillips and Sarnak (LPS) \cite{lubotzky1988} was the one chosen in \cite{breuckmann2021} to demonstrate the distance properties.

The LPS expander graphs are parametrised by two primes $p,q$, which specify that the graph contains $q(q^2-1)$ vertices, each of which has degree $s=p+1$. Hence, the number of edges is $q(q^2-1)(p+1)/2$. Being regular graphs, we can take their incidence matrix $I_0$. The graphs are expanding, and have an appropriate symmetry of order $l=q$. The null space has dimension $q(q^2-1)(p-1)/2$. We hold $p$ fixed in order to ensure the low density property, and scale $q$. When we apply the balanced product construction, prior to distance balancing, we find that there are $q(q^2-1)(p+3)/2$ physical qubits encoding at least $(q^2-1)(p-1)/2$ logical qubits. The $X$ and $Z$ distances, respectively, are $\order{q}$ and $\order{q^3}$. Distance balancing yields a code $\order{[[q^5,q^4,q^3]]}$, i.e.\ $\order{[[N,N^{4/5},N^{3/5}]]}$. This achieves $kd^2\sim N^2$. Crucially, the distance scales better than $\sqrt{N}$.

In terms of actually constructing the graph, the whole procedure is well explained in \cite{breuckmann2021}, and what we can add best is an example.
\begin{example}\label{ex:LPS}
Let $p=3$ and $q=5$.
We specify the vertices of a graph by correspondence to the $2\times 2$ invertible matrices where the elements are chosen mod $q$. This is the general linear group. However, we need to restrict to the projective general linear group, which means that all the matrices $A,2A,3A,4A$ map to the same group element/vertex. We choose the convention that the representative of this set's first non-zero entry is 1. For instance, one vertex, represented by $\identity$, corresponds to all the matrices
$$
\left\{
  \begin{bmatrix} 1 & 0 \\ 0 & 1 \end{bmatrix},\begin{bmatrix} 2 & 0 \\ 0 & 2 \end{bmatrix},\begin{bmatrix} 3 & 0 \\ 0 & 3 \end{bmatrix},\begin{bmatrix} 4 & 0 \\ 0 & 4 \end{bmatrix}
  \right\}.
$$

In order to determine the edges of the graph, we define a set of matrices $S_{3,5}$, with an edge between two vertices if their corresponding matrices $u,v$ satisfy $\sigma u=v$ for some $\sigma\in S_{3,5}$. To that end, we consider the $p+1=4$ integer solutions to the equation
$
a^2+b^2+c^2+d^2=3
$
where $a=0$ and $b=+1$:
$$
\!\!\!\!S_3\!=\!\{\!(0,1,1,1),\!(0,1,1,-1),\!(0,1,-1,1),\!(0,1,-1,-1)\!\}.
$$
Next, we pick an $x,y$ such that $x^2+y^2+1=0$ (mod $5$). For instance, $x=2,y=0$. This defines our set of matrices
$$
S_{3,5}=\left\{\begin{bmatrix} 2 & -3 \\ -1 & -2\end{bmatrix},\begin{bmatrix} 2 & -1 \\ -3 & -2\end{bmatrix},\begin{bmatrix} 2 & 3 \\ 1 & -2\end{bmatrix},\begin{bmatrix} 2 & 1 \\ 3 & -2\end{bmatrix}\right\}.
$$

The vertex $\identity$ is adjacent to the vertices $\sigma$ for $\sigma\in S_{3,5}$,
which have representatives
$$
\left\{\begin{bmatrix} 1 & 1 \\ 2 & 4\end{bmatrix},\begin{bmatrix} 1 & 2 \\ 1 & 4\end{bmatrix},\begin{bmatrix} 1 & 4 \\ 3 & 4\end{bmatrix},\begin{bmatrix} 1 & 3 \\ 4 & 4\end{bmatrix}\right\}.
$$

Consider the matrix
$$
H=\begin{bmatrix}1 & 1 \\ 0 & 1 \end{bmatrix}.
$$
Since $H^5=\identity$, for each vertex $u$, this creates an orbit of length $l=q=5$ comprising the vertices $A, AH,AH^2,AH^4,AH^4$. This permutation of the vertices defines the symmetry we require for the balanced product construction.

There are $q(q^2-1)=120$ vertices, each with degree 4 and, hence, there are $240$ edges.
The $120\times 240$ incidence matrix can directly be used to construct a balanced product code of 360 vertices. There are 26 logical qubits; the incidence matrix has a null space of dimension 120, but these are not all symmetric. Indeed, there are only 25 symmetric vectors (almost the minimum possible of $24=120/4$). On the left-hand block there is one additional vector (the all-ones). Some of these logical operators include the pairs
\begin{align*}
X_{L1}&=X_1X_2X_3X_4X_5 \\
Z_{L1}&=Z_1Z_2\ldots Z_{120} \\
X_{L2}&=X_{121}X_{122}\ldots X_{360} \\
Z_{L2}&=Z_{121}Z_{122}Z_{123}Z_{124}Z_{125}
\end{align*}
where the two length 5 operators are supported on a single orbit of $R,C$ respectively. The $X$ and $Z$ distances satisfy $d_X,d_Z\leq 5$. One can readily verify that these bounds are tight. Hence, the code is $[[360,26,5]]$. Ignoring the single logical qubit $\{X_{L1},Z_{L1}\}$ on the left-hand block would reduce to a $[[360,25,6,5]]$ subsystem code which is the analogue of what we have analysed in terms of distance in \cref{sec:gory}. The code details are explicitly described in \cite{kay2025a}, with the distance calculation being performed by \cite{kay2025}.
\end{example}

\subsection{Tanner Codes}

The original construction of the balanced product code \cite{breuckmann2021} does not only use expander matrices, but also uses a local code, combining the two in a Tanner code \cite{tanner1981}. We have shown in the previous arguments that this is in no way necessary. It does provide some advantage, however: if the local code has distance $d_l$, then the $X$ distance of the code is enhanced by a factor of approximately $d_l$ while the weights of some of the stabilizers are reduced (often reducing the maximum weight). For completeness, therefore, we include the construction here.

We start with $I_0$, an expander matrix with fixed row weight $s$ coming from, for example, an $s$-regular graph. It has symmetry operators $R_0$ and $C$, such that $R_0I_0=I_0C^T$. We now need to add an extra imposition on our choice of graph -- no two vertices in the same orbit of $R$ should share an edge of the graph.

To create a Tanner code, we supplement the expander with a `local' classical code $C_0$ that is $[s,k_l,d_l]$. The parity check of this local code has $r=s-k_l$ rows. For each row of $I_0$, we take the block of $s$ $1$s and replace it with the parity check matrix of $C_0$ (so, eventually, the $m$ rows are increased to $mr$ rows, leaving the number of columns fixed). In general for the Tanner code, each time you make a replacement, it is possible to pick any arbitrary permutation of the columns. This can lead to substantial variation in the code properties. However, we need our final matrix to have a symmetry. To do so, consider a vertex $v$ in the original graph (row of $I_0$). There are $s$ edges which connect to this vertex. Pick one of these edges, $e$. Imagine we assigned some column of $C_0$ to this particular edge. Then, when we need to assign columns of $C_0$ to $Rv$, we just have to do this consistently, so the same column that we assigned to $v,e$ must be assigned to $Rv,C^Te$. The only way this wouldn't be possible would be if, in looping through the orbit of $R$, we arrive at another vertex/edge pair that has already been assigned a column, i.e.\ if the graph has an edge between two vertices in the same orbit. We have already forbidden this. From this, we construct a new $I$ which has symmetry $R=R_0\otimes\identity_{r}$, i.e.\ $RI=IC^T$.

\begin{example}
Consider an initial matrix
$$
I_0=\begin{bmatrix} 1 & 1 & 0 & 0 & 1 & 0 \\ 0 & 0 & 1 & 1 & 0 & 1 \end{bmatrix}.
$$
Since this has row weights 3, we can consider introducing a local code of length 3 such as
$$
C_0=\begin{bmatrix} 1 & 1 & 0 \\ 0 & 1 & 1 \end{bmatrix}.
$$
We begin by applying to construction to the first row of $I_0$,
$$
\begin{bmatrix} 1 & 1 & 0 & 0 & 1 & 0 \end{bmatrix} \rightarrow 
\begin{tikzpicture}[baseline={([yshift=-2pt]m.center)}]
        \matrix [matrix of math nodes,left delimiter={[},right delimiter={]},inner sep=2pt] (m)
        {1 & 1 & 0 & 0 & 0 & 0 \\
0 & 1 & 0 & 0 & 1 & 0 \\};
\begin{scope}[on background layer]
 \node [fit={(m-1-1) (m-2-2)}, fill=blue!20,inner sep=1pt] {};
 \node [fit={(m-1-5) (m-2-5)}, fill=blue!20,inner sep=1pt] {};
 \end{scope}
\end{tikzpicture}
$$
where we have chosen a direct mapping without permuting the columns of $C_0$.

Now, we must preserve the symmetry of $I_0$ \footnote{For the sake of this example, we are not worrying about the two symmetry operators having different orders.}
\[
R_0=\begin{bmatrix} 0 & 1 \\ 1 & 0 \end{bmatrix},\qquad
C=\begin{tikzpicture}[baseline={([yshift=-2pt]m.center)}]
        \matrix [matrix of math nodes,left delimiter={[},right delimiter={]},inner sep=2pt,ampersand replacement=\&] (m)
        {0 \& 0 \& 0 \& 0 \& 0 \& 1 \\
 0 \& 0 \& 1 \& 0 \& 0 \& 0 \\
 1 \& 0 \& 0 \& 0 \& 0 \& 0 \\
 0 \& 0 \& 0 \& 0 \& 1 \& 0 \\
 0 \& 0 \& 0 \& 1 \& 0 \& 0 \\
 0 \& 1 \& 0 \& 0 \& 0 \& 0 \\};
\begin{scope}[on background layer]
 \node [fit={(m-2-3)}, fill=green!20,inner sep=1pt] {};
 \end{scope}
\end{tikzpicture}.
\]
This means that for the first vertex, where the first column got replaced with $\begin{bmatrix} 1 \\ 0 \end{bmatrix}$, in the second row it is the last column that requires a similar replacement. Ultimately, then, we will have
$$
I=\begin{tikzpicture}[baseline={([yshift=-2pt]m.center)}]
        \matrix [matrix of math nodes,left delimiter={[},right delimiter={]},inner sep=2pt] (m) {
1 & 1 & 0 & 0 & 0 & 0 \\
0 & 1 & 0 & 0 & 1 & 0 \\
0 & 0 & 1 & 0 & 0 & 1 \\
0 & 0 & 1 & 1 & 0 & 0 \\};
\begin{scope}[on background layer]
 \node [fit={(m-1-2) (m-2-2)}, fill=green!20,inner sep=1pt] {};
 \node [fit={(m-3-3) (m-4-3)}, fill=green!20,inner sep=1pt] {};
 \end{scope}
 \end{tikzpicture},
$$
invariant under the action of $C$ and $R=R_0\otimes\identity_2$.
\end{example}

\section{Summary}

In this paper, we have reviewed the construction of the balanced product code from \cite{breuckmann2021}, taking a very different perspective from the original presentation, which we hope is illuminating and will facilitate further study, such as the capabilities of transversal gates \cite{leitch}. We have concentrated on an extreme case of the construction and, in that case, have given the details of how the distance of the code can be seen to be growing more rapidly than $\sqrt{N}$. This is what made this result an essential building block on the road towards good quantum LDPC codes.

Ideally, we were aiming to beat the threshold of $kd^2\sim n^2$. Given the distance balancing, there is an equivalent threshold prior to to distance balancing, of $kd_Xd_Z\sim n^2$. For the balanced product, we have that $k\sim n/l$, $d_Z\sim n$ and $d_X\sim l$. Thus, whatever expander code we pick, no matter the degree of symmetry $l$, we will match, but not exceed, the threshold. However, it is worth noting that we have only bounded $k\geq k_0/l$ where $k_0\sim n$. There is potential to find cases where the performance far exceeds that, breaking the threshold. That said, all the examples that we explore in \cite{kay2025a} saturate the $k=k_0/l$ bound.

The essential elements in the construction require a binary matrix with some symmetry, from which we can directly express the parity-check matrices of a quantum code, \cref{def:balanced_product}. While not necessary, one of the distances can be improved by utilising the Tanner code construction (note, however, that this increases the number of rows in the initial parity-check matrix, which increases the number of physical qubits in the balanced product). The result is a code with a large number of logical qubits, and high, but asymmetric, $X$ and $Z$ distances. The effect of this asymmetry can be mitigated by the distance balancing procedure of \cref{def:distance_balance}.

This work was supported by the Engineering and Physical Sciences Research Council [grant number EP/Y004507/1]. We would like to thank the research group of Dan Browne at UCL for useful conversations.

\newpage

\begin{thebibliography}{35}
\bibitem [{a()}]{a}%
  \BibitemOpen
  \bibinfo {note} {For the sake of this example, we are not worrying about the two symmetry operators having different orders.}\BibitemShut {Stop}%
\end{thebibliography}
%

\clearpage
\onecolumngrid
\appendix
\section{Where's the Product?}\label{sec:missing_product}

Throughout this paper, the product part of the balanced product code has been conspicuous in its absence. Where did it go? We have only considered the product between one (mostly) arbitrary graph, and a fixed graph (the cycle of length equal to the orbit length of the first graph), corresponding to the primary focus of \cite{breuckmann2021}. As such, we have been able to immediately calculate the effect of the product, and simply start from a point where the product has already been evaluated. In the following, we will justify our assertions, while giving more general results and seeing how they connect with, for example, the hypergraph product codes.

\begin{theorem}
Let $I_X\in\{0,1\}^{l\times n_X,l\times m_X}$ and $I_Y\in\{0,1\}^{n_Y\times l,m_Y\times l}$ be two matrices that satisfy
\[
(R\otimes\identity_{n_X})I_X=I_X(R\otimes\identity_{m_X}), \qquad
(\identity_{n_Y}\otimes R)I_Y=I_Y(\identity_{m_Y}\otimes R)    
\]
where $R$ is an $l\times l$ permutation matrix. The balanced product defines the pair of parity check matrices
\[
H_X=\begin{bmatrix} \identity_{n_Y}\otimes I_X & I_Y\otimes \identity_{n_X}\end{bmatrix}, \qquad
H_Z=\begin{bmatrix} I_Y^T\otimes\identity_{m_X} & \identity_{m_Y}\otimes I_X^T \end{bmatrix}.  
\]
\end{theorem}
There are some special cases worth considering. In the absence of a symmetry, i.e.\ $l=1$, then $I_X$ and $I_Y$ act on distinct spaces, and the construction is exactly that of the hypergraph product code \cite{tillich2014}. Taken to the opposite extreme of maximising the intersection of the spaces, we let $n_Y=m_Y=1$. If we then let $I_Y$ be the repetition code $\identity+R$, we recover our claimed form of the balanced product code, \cref{def:balanced_product}.

The commutation relation $H_XH_Z^T=0$ may not be immediately apparent in the above. However, note that, for example, we can write $I_X=\sum_{i,j}M_{ij}\otimes\ket{i}\bra{j}$, so it must be true that $RM_{ij}=M_{ij}R$ for all $i,j$. This commutation relation means that we can rewrite $M_{ij}=f_{ij}(R)$ for some polynomial $f_{ij}(x)$. We thus have that $I_X=\sum_{i,j}f_{ij}(R)\otimes\ket{i}\bra{j}$ and $I_Y=\sum_{i,j}\ket{i}\bra{j}\otimes g_{ij}(R)$. Hence, it must be the case that $$(\identity_{n_Y}\otimes I_X)(I_Y\otimes \identity_{m_X})=(I_Y\otimes \identity_{n_X})(\identity_{m_Y}\otimes I_X),$$ which is all we need.

\begin{proof}
Consider two graphs $X$ and $Y$, both of which have a permutation symmetry so that all vertices/edges partition into orbits of the same length $l$. We can choose to label any vertex (or edge) based on which orbit it is in, and the position in the orbit under the action of some $R$ relative to a representative element of the orbit. We choose to order these labels differently for the two graphs, i.e.\ $X$ is labelled by (shift in orbit,orbit) while $Y$ is labelled by (orbit, shift in orbit). Hence, by construction, we have
\[
(R\otimes\identity_{n_X})I_X=I_X(R\otimes\identity_{m_X}), \qquad
(\identity_{n_Y}\otimes R)I_Y=I_Y(\identity_{m_Y}\otimes R).    
\]
The aim of the balanced product is to combine the two graphs under the Cartesian product, gaining something by reducing the ``shift in orbit'' parts of the space that each has into a single common space.

Specifically, label a set of vertices
$$
z_{a,b,c}=\left\{y_{a,-i}\otimes x_{i+b,c}\right\}_{i=0}^{l-1}
$$
where $y_{a,b}$ is a vertex of $Y$ with an offset of $b$ inside orbit $a$ and $x_{b,c}$ is a vertex of $X$ with an offset of $b$ inside orbit $c$. We say that the new graph $Z$ has an edge if either $X$ had an edge between one of its members (with the same $y$) or its $Y$ had an edge between one of its members with the same $x$. We claim that the incidence matrix of this new graph may be written as $H_X$.

In particular, assume that $X$ has an edge between $x_{a,b}$ and $x_{c,d}$. For every orbit $p$ of the $Y$ vertices, $y_{p,0}\otimes x_{a,b}$ is in the set $z_{p,a,b}$ while $y_{p,0}\otimes x_{c,d}$ is in the set $z_{p,c,d}$. So there must be an edge between $z_{p,a,b}$ and $z_{p,c,d}$ (note that there are actually $l$ edges, as we permute through every element of each set, but we only count it once). Hence, we get a set of edges described in the form $\identity\otimes I_X$ derived from the edges in $X$.

Similarly, assume that $Y$ has an edge between $y_{a,b}$ and $y_{c,d}$. For every orbit $p$ of the $X$ vertices, $y_{a,b}\otimes x_{0,p}$ and $y_{c,d}\otimes x_{0,p}$ must be joined by an edge. Permuting through the orbit, these are in the same $z$ sets as $y_{a,0}\otimes x_{b,p}$ and $y_{c,0}\otimes x_{d,p}$ respectively. Thus, these are members of $z_{a,b,p}$ and $z_{c,d,p}$ respectively, so there must be an edge between these. Overall, this must give a set of edges in the incidence matrix of
$I_Y\otimes\identity$.

Simple counting shows that we must have all the edges. In the full Cartesian product of the graph, we would have $E_XV_Y+E_YV_X$ edges, where the graphs had $E_X$ edges and $V_X$ vertices. When forming the graph $Z$, we have retained $n_XV_Y+n_YV_X=(E_XV_Y+E_YV_X)/l$ edges, having already noted that every edge in this graph corresponds to $l$ in the Cartesian product.
\end{proof}

\end{document}